%% file: ITW_26RM.tex
\documentclass[conference]{IEEEtran}
\IEEEoverridecommandlockouts
\usepackage{cite}
\usepackage{amsmath,amssymb,amsfonts}
\usepackage{algorithmic}
\usepackage{algorithm}
\usepackage{graphicx}
\usepackage{textcomp}
\usepackage{xcolor}
\def\BibTeX{{\rm B\kern-.05em{\sc i\kern-.025em b}\kern-.08em
    T\kern-.1667em\lower.7ex\hbox{E}\kern-.125emX}}
    
\usepackage{bm} 
\usepackage{amsthm} 

\usepackage{multicol}
\usepackage{multirow}
\usepackage{threeparttable}
\usepackage{dsfont}
\usepackage{makecell}
\usepackage{booktabs}
\usepackage{tikz}
\usepackage{pgfplots}
\pgfplotsset{compat=1.18}
\usepackage{flushend}

\newtheoremstyle{mythestyle}
  {3pt}   
  {3pt}   
  {\normalfont} 
  {}   
  {\itshape} 
  {:}   
  {0.5em} 
  {\thmname{#1}\thmnumber{ #2}\thmnote{ \textit{(#3)}}} 

\theoremstyle{mythestyle}
\newtheorem{theorem}{Theorem}  
\newtheorem{lemma}{Lemma}  

\newtheorem{definition}{Definition}

\newcommand{\Ft}{\mathbb{F}_2}

\newcommand{\Fq}{\mathbb{F}_q}
\newcommand{\Fqm}{\mathbb{F}_q^m}

\newcommand{\AG}{\operatorname{AG}}
\newcommand{\PG}{\operatorname{PG}}

\newcommand{\cG}{\mathcal{G}}
\newcommand{\cP}{\mathcal{P}}

\newcommand{\qbinom}[2]{\genfrac{[}{]}{0pt}{}{#1}{#2}_q}
\DeclareMathOperator{\rank}{rank}

\begin{document}

\title{Lifted Gabidulin Construction for LDPC Representations of Finite Geometry Codes}

\author{\IEEEauthorblockN{Yifei Shen}
\IEEEauthorblockA{
\textit{TCL, EPFL, Switzerland}\\
Email: yifei.shen@epfl.ch
}
\and
\IEEEauthorblockN{Andreas Burg}
\IEEEauthorblockA{
\textit{TCL, EPFL, Switzerland}\\
Email: andreas.burg@epfl.ch
}
}

\maketitle

\begin{abstract}
Finite geometry (FG) codes combine the algebraic properties of classical block codes with the iterative belief propagation (BP) decoding ability of low-density parity-check~(LDPC) codes. 
However, exploiting both advantages in practice is hindered by the fact that the standard incidence matrix between $(\mu+1)$-flats and points is dense and contains many short cycles for any flat dimension $\mu\geq 1$. 
In this work, we propose to sparsify the decoding matrix based on pencil selection, formulated as a constant-dimension subspace packing problem and solved explicitly using lifted Gabidulin codes. 
For both affine and projective geometries, sparse parity-check matrices are constructed and verified for FG codes of lengths up to $1024$.
Simulations on four FG codes show no visible error floor and around $0.5$~dB gain over corresponding 5G LDPC codes at a block error rate of $10^{-7}$.
\end{abstract}


\section{Introduction}\label{sec:intro}
Coding with short blocklengths has emerged as a key problem in next-generation communication systems~\cite{miao2024trends}.
In this regime, the asymptotic advantages of modern low-density parity-check~(LDPC) and polar codes diminish, while classical algebraic codes remain competitive due to their large minimum distances~\cite{shirvanimoghaddam2018short}. Compared to practical Bose-Chaudhuri-Hocquenghem (BCH) codes, finite geometry (FG) codes~\cite{rudolph1964,smith1967,weldon1967euclidean} provide a compelling trade-off between complexity and performance due to their reasonably large minimum distances and simple majority-logic decoding~\cite{LinCostello2004}. These properties motivated extensive work on FG codes in the late 1960s and 1970s~\cite{delsarte1969geometric,hartmann1972generalized,lin1973multifold,hartmann1974structure}. Beyond their historical success with hard-decision decoding, their rich automorphism groups and highly regular geometric structures also make FG codes amenable to modern iterative soft-decision decoding~\cite{delsarte1970generalized,huang2009two,chen2018finite,yu2025ordered}.

Similar to LDPC codes, FG codes also experienced a re-discovery three decades after their inception. Kou \emph{et al.}~\cite{kou2002low} demonstrated that the point-line incidence structure of FG codes naturally yields $4$-cycle-free Tanner graphs, which allows these FG-LDPC codes to be efficiently decoded by the belief propagation (BP) algorithm. This code family was further broadened by generalizing the geometric incidence rule and requiring the dimensions of the flats indexing rows and columns to differ by exactly one~\cite{tang2005codes}. Compared to heuristically constructed LDPC codes, FG-LDPC codes possess large minimum distances and natural cyclic or quasi-cyclic structures~\cite{kamiya2006quasi,kamiya2007high}. Moreover, analyses of the trapping sets and pseudo-codewords of FG-LDPC codes~\cite{smarandache2007pseudo,liu2012smallest} reveal that they are unlikely to suffer from severe error floors, making them particularly suitable for the short-blocklength regime. However, the availability of FG codes that possess naturally $4$-cycle-free Tanner graphs is strictly limited, while incidence matrices of most short FG codes contain numerous $4$-cycles.

To reduce the complexity of dense Tanner graphs, we can introduce auxiliary variable nodes (AVNs) to split check nodes and therefore sparsify the graph~\cite{mackay2000relationships,yedidia2002generating,sankaranarayanan2005iterative,kumar2005graphical,shen2025toward}. 
Our recent work introduced a general methodology for extracting sparse parity-check matrices for arbitrary short linear block codes. This approach progressively selects parity-check row blocks (PCRBs) from a vast search space of low-weight dual codewords, which results in structured, sparse PCMs (ssPCMs)~\cite{shen2025belief}. Iterative BP decoding on ssPCMs shows significant performance gains over the original PCMs, as verified for most short codes with lengths up to $128$. However, the search space of candidate parity checks grows rapidly with blocklength, making the matrix extraction procedure eventually prohibitively expensive.

In this paper, we propose a geometric approach to derive sparse LDPC representations of FG codes. Our construction relies on the notion of a pencil, which is the set of all $(\mu+1)$-flats containing a fixed $\mu$-flat, which yields a cycle-free PCRB after introducing one AVN. Therefore, the construction of an ssPCM for FG codes reduces to a pencil selection problem. Inspired by the constant-dimension codes from network coding theory~\cite{Silva2008}, we solve this problem using the lifted Gabidulin construction~\cite{Gabidulin1985}. For both affine and projective geometries, we provide explicit constructions and verify rank preservation for FG codes of lengths up to $1024$. Simulations on flooding BP decoding of four FG codes show no visible error floor and around $0.5$~dB gain over corresponding 5G LDPC codes at a block error rate~(BLER) of $10^{-7}$. 
\section{Finite Geometry}\label{sec:prelim} 

\emph{Notation:} Let $m$ and $s$ be positive integers, and let $\Fq$ denote the finite field of order~$q=2^s$. We consider two families of binary FG codes associated with affine geometry $\AG(m,q)$ and projective geometry $\PG(m,q)$ over~$\Fq$. For a fixed integer~$\mu$ with $0\le\mu\le m-1$, the FG($m,s,\mu$) code is defined as the null space over $\mathbb{F}_2$ of the point--$(\mu+1)$-flat incidence matrix.
We write $\qbinom{x}{y}$ for the Gaussian binomial coefficient over~$\Fq$. For two subspaces $U,V$, their distance is defined as 
$d_{\rm S}(U,V)=\dim(U)+\dim(V)-2\dim(U\cap V)$. The $\Ft$-row span of a binary matrix $\mathbf{H}$ is denoted as $\mathrm{row}_{\Ft}(\mathbf{H})$. 

\subsection{Incidence Matrices of Finite Geometry}
An incidence matrix is indexed by $(\mu+1)$-flats~(rows) and points~(columns), with entry $1$ whenever a point belongs to the corresponding flat. Each row is an indicator vector of a $(\mu+1)$-flat, with weight equal to the number of points in that flat. These rows are minimum-weight dual codewords of the corresponding FG code~\cite{delsarte1970generalized}. 

The affine geometry $\AG(m,q)$ has point set $\cP=\Fq^m$. An affine $(\mu+1)$-flat is a coset $F=\bm{z}+U$, where $\bm{z}\in\Fq^m$, $U\le\Fq^m$, and $\dim(U)=\mu+1$. A $(\mu+1)$-flat contains $q^{\mu+1}$ points. Hence the AG incidence matrix~$\mathbf{H}_{\mathrm{AG}}$ has $q^m$ columns indexed by points and $q^{m-\mu-1}\qbinom{m}{\mu+1}$ rows indexed by all affine $(\mu+1)$-flats, each of weight $w_r=q^{\mu+1}$. When $q=2$, codes defined by $\AG (m,q)$ are Reed-Muller codes~\cite{delsarte1970generalized}.

The projective geometry $\mathrm{PG}(m,q)$ can be viewed as the completion of $\mathrm{AG}(m,q)$ by adding a hyperplane at infinity. Its point set
$\cP$ equals the one-dimensional subspaces of~$\Fq^{m+1}$. 
In this paper, we use $\mathbb{P}(W)$ to represent the projectivization of a subspace~$W$. If $\dim(W)=\mu+1$, then $\mathbb{P}(W)$ is a projective $\mu$-flat in $\PG (m,q)$. 
The PG incidence matrix $\mathbf H_{\mathrm{PG}}$ has $\qbinom{m+1}{1}$
columns indexed by points and 
$\qbinom{m+1}{\mu+2}$
rows indexed by all projective $(\mu+1)$-flats, each of weight
$w_r=\qbinom{\mu+2}{1}$. 

\subsection{Pencils in Finite Geometry}
Following the terminology of geometry, a pencil is a family of geometric objects with a common property~\cite{halsted1896synthetic}. In our context, a pencil refers to the set of all $(\mu+1)$-flats containing a given $\mu$-flat~$\Pi$.
\begin{definition}[Pencil]\label{def:pencil}
Let $\mathcal{F}_{\mu+1}$ denote the set of all $(\mu+1)$-flats in either $\AG(m,q)$ or $\PG(m,q)$. For a $\mu$-flat $\Pi$, the \emph{pencil} centered at~$\Pi$ is
\begin{equation}\label{eq:pencil}
\cG(\Pi) \triangleq
\{F \in \mathcal{F}_{\mu+1}: \Pi \subseteq F\}.
\end{equation}
\end{definition}

\begin{lemma}\label{lem:pencil_prop}
Let $\Pi$ be a $\mu$-flat, then 
\begin{equation}
|\cG(\Pi)| = J \triangleq \qbinom{m-\mu}{1} = (q^{m-\mu} - 1)/(q-1).
\end{equation}
For any two distinct $F_1, F_2 \in \cG(\Pi)$, $F_1 \cap F_2 = \Pi$. 
\end{lemma}

\begin{proof} 
  For $\mathrm{AG}(m,q)$, denote the kernel as $\Pi=\bm{z}+W$, where $W\le \mathbb F_q^m$ and $\dim(W)=\mu$. Since $\bm{z}\in\Pi$, any affine $(\mu+1)$-flat containing~$\Pi$ admits $\bm{z}$ as a coset representative, hence has the form $\bm{z}+U$, where $W<U\le \mathbb F_q^m$ and $\dim(U)=\mu+1$. For $\mathrm{PG}(m,q)$, write $\Pi=\mathbb P(W)$ with $W\le\mathbb F_q^{m+1}$ and $\dim(W)=\mu+1$. Every projective $(\mu+1)$-flat containing $\Pi$ is $\mathbb P(U)$, where $W<U\le\mathbb F_q^{m+1}$ and $\dim(U)=\mu+2$. 
  
  In both cases, such flats are in one-to-one correspondence with the one-dimensional subspaces of the quotient space ($\mathbb F_q^m/W$ for $\AG(m,q)$ and $\mathbb F_q^{m+1}/W$ for $\PG(m,q)$), whose number is $J\!=\!\qbinom{m-\mu}{1}$. Since distinct one-dimensional subspaces of the quotient space intersect trivially, the corresponding $(\mu+1)$-flats intersect exactly in the common kernel $\Pi$.
\end{proof}

The number of possible kernels, i.e., the total number of centered full pencils, is $q^{m-\mu}\qbinom{m}{\mu}$ in $\AG(m,q)$ and $\qbinom{m+1}{\mu+1}$ in $\PG(m,q)$, respectively.

\section{Matrix Sparsification Based on Pencils}\label{sec:avn}
As shown in~\cite{tang2005codes}, BP decoding on the incidence matrix of an FG code can provide good soft-decision performance. When $\mu=0$, since two distinct lines intersect in at most one point, the point-line incidence matrix is $4$-cycle-free. However, the incidence matrix is not a suitable BP representation when $\mu\geq 1$. Although its rows are minimum-weight dual codewords, they are highly redundant and induce a dense Tanner graph with many $4$-cycles. Moreover, with $q^m\qbinom{m}{\mu+1}$ edges for $\AG(m,q)$ and $\qbinom{m+1}{\mu+2}\qbinom{\mu+2}{1}$ edges for $\PG(m,q)$, the computational complexity of the BP algorithm is very high. 

\subsection{From Pencil to Parity-Check Row Block}
Our prior work proposed to sparsify dense parity-check matrix based on PCRBs~\cite{shen2025belief}. A PCRB is a collection of parity checks whose supports share a common subset of coded-bit columns. For each PCRB, one AVN is added to split each original parity check by replacing the common subset with the AVN. The remaining supports of the checks in a PCRB are pairwise disjoint, then the Tanner subgraph induced by the block is \mbox{$4$-cycle-free}. 
For a given short code, the extraction of PCRBs from a large set of low-weight dual codewords becomes increasingly expensive as the blocklength grows.

However, for FG codes, the pencil structure provides a natural candidate for PCRBs. For a $\mu$-flat $\Pi_i$, every \mbox{$F\in\cG(\Pi_i)$} contains the same kernel $\Pi_i$, and the $J$ incidence rows indexed by $\cG(\Pi_i)$ form the original rows of a PCRB. We introduce one AVN~$v^\prime$ and impose the defining check
\begin{equation}\label{eq:avn}
v^\prime \oplus \bigoplus_{P\in\Pi_i}v_{P}=0,
\end{equation}
where $v_P$ represents the coded bit indexed by point~$P$.
For each $F$, the original check is replaced by
\begin{equation}\label{eq:avn_check}
v^\prime \oplus \bigoplus_{P\in F\setminus\Pi_i}v_{P}=0 .
\end{equation}
The replacement preserves exactly the same constraints on the original coded-bit variables. 
The sparsification results from removing the common kernel from every row in the pencil. For $\AG(m,q)$, each rewritten incidence row has weight $q^\mu(q-1)+1$ and the newly added row has weight $q^\mu+1$. For $\PG(m,q)$, each rewritten incidence row has weight $q^{\mu+1}+1$ and the newly added row has weight $\qbinom{\mu+1}{1}+1$. 

Since distinct rows in a full pencil intersect only on $\Pi_i$, their rewritten supports share only $v^\prime$ and the local pencil block is therefore cycle-free. As a result, the AVN is connected to every coded bit through one of the $J+1$ check rows. This high-degree AVN therefore receives reliable channel information from many coded bits after the first iteration. Moreover, by Theorem~1 of~\cite{shen2026hdpc}, the size of any absorbing set~\cite{dolecek2009analysis} involving an AVN admits a lower bound that grows with the PCRB height. Since full pencils attain the maximum possible height, the ssPCM for FG codes is free of overly small absorbing sets.
 
\subsection{Reformulation of ssPCM Construction as Pencil Selection}
Given a full pencil that induces one cycle-free PCRB, it remains to select enough pencils so that the resulting ssPCM preserves the original code. We formalize this requirement as the pencil selection problem.

\begin{definition}[Pencil Selection Problem]\label{def:problem}
For an incidence matrix $\mathbf{H}$, select kernels $\Pi_1,\ldots,\Pi_B$ such that 
\begin{equation}\nonumber
\label{eq:pencil_selection}
\begin{aligned}
\!\!{\rm (P1)}\;
&\mathcal G(\Pi_i)\cap \mathcal G(\Pi_j)=\varnothing,
\; \forall i\ne j,
\\ 
\!\!{\rm (P2)}\;
&\operatorname{rank}_{\mathbb F_2}(\mathbf H_{\rm p})
=
\operatorname{rank}_{\mathbb F_2}(\mathbf H),
\;
\mathbf H_{\rm p}\triangleq
\begin{bmatrix}
\mathbf H[\mathcal G(\Pi_1),:]\\
\vdots\\
\mathbf H[\mathcal G(\Pi_B),:]
\end{bmatrix}\!.
\end{aligned}
\end{equation}
\end{definition}
Here, (P1) ensures that no incidence row is reused in different PCRBs, and (P2) guarantees that the selected incidence rows generate the same binary row space as $\mathbf{H}$. By adding one AVN per selected pencil, we obtain a $(J+1)B \times (|\mathcal{P}|+B)$ ssPCM with row weight~$\tilde{w}_r$ at most $q^\mu(q-1)+1$ for $\AG(m,q)$ and $q^{\mu+1}+1$ for $\PG(m,q)$. 

Unfortunately, the pencil selection problem is nontrivial, since the number of possible kernels grows exponentially with $m$ and $\mu$. In Theorems~\ref{lem:disjoint_cond} and~\ref{lem:pg_disjoint_cond} below, we reformulate the disjointness condition (P1) as a constant-dimension subspace packing problem on the kernel directions. The explicit construction of pencil sets satisfying (P1), together with the offsets needed to fulfill (P2), will be introduced in Section~\ref{sec:construction}.

\begin{theorem}\label{lem:disjoint_cond}
Let~$\Pi_i = \bm{z}_i + W_i$ and $\Pi_j = \bm{z}_j + W_j$ be two affine \mbox{$\mu$-flats} in $\AG (m,q)$ with $1\leq \mu \leq m-2$, where $W_i,W_j\le\mathbb F_q^m$ and $\dim(W_i)=\dim(W_j)=\mu$. 
Then, $\cG(\Pi_i)\cap \cG(\Pi_j) = \varnothing$,
if and only if one of the following conditions holds:
\begin{subequations}
\renewcommand{\theequation}{\theparentequation-\arabic{equation}}
\begin{align}
  &\dim(W_i \cap W_j) \leq \mu - 2,\label{eq:ag_disjoint_subspace}\\
  \!\!\!\!\!\!\!\text{or}\;&\dim(W_i \cap W_j) = \mu - 1\;\mathrm{and}\;
  \bm{z}_i-\bm{z}_j \notin W_i + W_j.\label{eq:ag_disjoint_offset}
\end{align}
\end{subequations}
\end{theorem}

\begin{proof}
Assume that there exists a~$(\mu+1)$-flat~$F = \bm{y} + U$ such that~$F\in \cG(\Pi_i)\cap \cG(\Pi_j)$. Since $F$ contains both $\Pi_i$ and~$\Pi_j$, we have
  $\bm{z}_i-\bm{y},\bm{z}_j-\bm{y}\in U$, which leads to $W_i,W_j \subseteq U$. 
Therefore
\begin{equation}\label{eq:cond1}
  W_i+W_j \subseteq U \quad\mathrm{and}\quad \bm{z}_i-\bm{z}_j \in U.
\end{equation}
Conversely, if there exists a~$(\mu+1)$-dimensional subspace~$U$ such that~\eqref{eq:cond1} holds, then the affine flat $\bm{z}_i+U$ contains $\Pi_i$. It also contains $\Pi_j$, as $\bm{z}_i-\bm{z}_j\in U$ and $W_j\subseteq U$. Hence, the two pencils have a common element if and only if there exists such a subspace $U$.

Since $W_i+W_j\subseteq U$, we have
\begin{equation}\label{eq:cond2}
\dim(W_i+W_j)=2\mu-\dim(W_i\cap W_j)\leq \mu+1.
\end{equation}

If~$\dim(W_i\cap W_j)\leq \mu-2$, then~\mbox{$\dim(W_i+W_j)\geq \mu+2$}, which contradicts~\eqref{eq:cond2}. Therefore, no common~$(\mu+1)$-flat exists, and $\cG(\Pi_i) \cap \cG(\Pi_j) = \varnothing$.

If~$\dim(W_i\cap W_j)=\mu-1$, then~$\dim(W_i+W_j)=\mu+1$. Hence every common~$(\mu+1)$-flat of~$\cG(\Pi_i)$ and~$\cG(\Pi_j)$ must have direction space~$U=W_i + W_j$. Based on~\eqref{eq:cond1}, such a common flat exists if and only if
\begin{equation}
\bm{z}_i - \bm{z}_j \in U = W_i + W_j.
\end{equation}
Equivalently, in this case the two pencils are disjoint if and only if~$\bm{z}_i - \bm{z}_j \notin W_i + W_j$.

Finally, if $\dim(W_i\cap W_j)=\mu$, then $W_i=W_j$. Since $\mu\leq m-2$, there exists a $(\mu+1)$-dimensional subspace $U$ containing both $W_i$ and $\bm{z}_i-\bm{z}_j$, so the two pencils are not disjoint. Combining the three cases, the two pencils are disjoint if and only if one of~\eqref{eq:ag_disjoint_subspace} or~\eqref{eq:ag_disjoint_offset} holds.
\end{proof}

\begin{theorem}\label{lem:pg_disjoint_cond}
Let~$\Pi_i\!=\!\mathbb P(W_i)$ and $\Pi_j\!=\!\mathbb P(W_j)$ be two projective \mbox{$\mu$-flats} in $\PG(m,q)$ with $1\leq \mu \leq m-2$, where
$W_i,W_j\le \mathbb F_q^{m+1}$ and
$\dim(W_i)=\dim(W_j)=\mu+1$. 
Then, $\cG(\Pi_i)\cap \cG(\Pi_j)=\varnothing$, 
if and only if
\begin{equation}\label{eq:pgcond}
\dim(W_i\cap W_j)\le \mu-1.
\end{equation}
\end{theorem}

\begin{proof}
A common projective $(\mu+1)$-flat exists if and only if there is
a subspace $U\le \mathbb F_q^{m+1}$ with $\dim(U)\!=\!\mu+2$ such that
\[
W_i+W_j\subseteq U .
\]
Equivalently, such a common flat exists if and only if
\[
\dim(W_i+W_j)\le \mu+2 .
\]
Since $\dim(W_i+W_j)=
2(\mu+1)-\dim(W_i\cap W_j)$, the two pencils have a common element if and only if \mbox{$\dim(W_i\cap W_j)\ge \mu$}. 
Equivalently, the two pencils are disjoint if and only if $\dim(W_i\cap W_j)\le \mu-1$.
\end{proof}

Theorems~\ref{lem:disjoint_cond} and~\ref{lem:pg_disjoint_cond}
reduce the disjoint-pencil constraint (P1) to a condition on subspace distance. Both~\eqref{eq:ag_disjoint_subspace} for $\AG (m,q)$ and~\eqref{eq:pgcond} for $\PG (m,q)$ are equivalent to 
$d_{\rm S}(W_i,W_j)\ge 4$, since the kernel directions are $\mu$-dimensional in $\AG (m,q)$ and $(\mu+1)$-dimensional in $\PG (m,q)$. The affine case admits an additional refinement when $d_{\rm S}(W_i,W_j)=2$, disjoint pencils can still be obtained by choosing the offsets such that $\bm{z}_i-\bm{z}_j\notin W_i+W_j$, as stated in~\eqref{eq:ag_disjoint_offset}. 

A natural approach to enforce $d_\mathrm{S}\geq 4$ is to draw the kernel directions from a constant-dimension subspace code. Such a code can be obtained by lifting a rank-metric code, which has been widely employed in network coding~\cite{Silva2008}. Among rank-metric codes, Gabidulin codes are maximum rank distance (MRD) codes and allow for an explicit algebraic construction.
\section{Pencil Selection Based on Gabidulin Codes}\label{sec:construction}
\subsection{Lifted Gabidulin Codes}
We briefly recall the lifted Gabidulin construction used to generate constant-dimension subspace codes. Let $1\leq n \leq M$. For each matrix $\mathbf{A}\in\mathbb{F}_q^{n\times M}$, define the lifted $n$-dimensional subspace of $\mathbb{F}_q^{n+M}$ by
\begin{equation}\label{eq:lift}
  \mathcal{L}(\mathbf{A})\triangleq \mathrm{rowspan}[\mathbf{I}_n\; \mathbf{A}] \subseteq \Fq^{n+M}.
\end{equation}
For any $\mathbf{A}_i,\mathbf{A}_j\in\mathbb{F}_q^{n\times M}$, we have
\begin{equation}\label{eq:subspace_distance}
  d_{\rm S}(\mathcal{L}(\mathbf{A}_i),\mathcal{L}(\mathbf{A}_j))=2\rank_{\Fq}(\mathbf{A}_i-\mathbf{A}_j).
\end{equation}
  \vspace{-2.5pt}
Hence, a rank-metric code with minimum rank distance $\delta$ is lifted to a constant-dimension subspace code with minimum subspace distance $2\delta$. Gabidulin codes are MRD codes and meet the Singleton bound with equality~\cite{Gabidulin1985}. 

Let $\mathcal{C}_\mathrm{Gab}$ be an $(n,k,\delta)$ Gabidulin code over $\mathbb{F}_{q^M}$, where $\delta \!=\! n \!-\! k \!+\! 1$. For each vector $\bm{u}\!=\!(u_0,u_1,\ldots,u_{k-1})\in\mathbb{F}^k_{q^M}$, define the $q$-linearized polynomial $f_{\bm{u}}(x)=\sum_{i=0}^{k-1} u_i x^{q^{i}}$. Choose $g_1, g_2, \ldots, g_n \!\in\! \mathbb{F}_{q^M}$ linearly independent over $\Fq$. 
 The corresponding codeword is given by
  \vspace{-2.5pt}
 \begin{equation}
 \bm{c}=(f_{\bm{u}}(g_1), f_{\bm{u}}(g_2), \ldots, f_{\bm{u}}(g_n))\in\mathbb{F}^n_{q^M}.
 \end{equation}
 By expanding each symbol of $\bm{c}$ with respect to a fixed $\mathbb{F}_q$-basis of $\mathbb{F}_{q^M}$, we obtain a matrix $\mathbf{A}\!\in\!\mathbb{F}_q^{n\times M}$. The cardinality of the resulting code is 
 \vspace{-0.5mm}
\begin{equation}
  |\mathcal{C}_\mathrm{Gab}| = q^{kM}.
\end{equation}
\vspace{-0.5mm}
Except in the affine $\mu=1$ boundary case, we fix $\delta=2$, so that the lifted Gabidulin code satisfies $d_{\rm S}(\mathcal{L}(A_i),\mathcal{L}(A_j))\ge 4$ for all distinct $\mathbf{A}_i,\mathbf{A}_j\in\mathcal{C}_\mathrm{Gab}$, which provides a sufficient condition for pencil disjointness in Theorems~\ref{lem:disjoint_cond} and~\ref{lem:pg_disjoint_cond}.

\vspace{-1mm}
\subsection{Pencil Selection for Affine Geometry Codes}
In this section, we give an explicit construction of pencils for $\AG(m,q)$ based on lifted Gabidulin codes. Identify $\Fqm \cong \Fq^\mu \times \Fq^{m-\mu}$. For a matrix~$\mathbf{A}\in\mathbb{F}_q^{\mu\times(m-\mu)}$, define the $\mu$-dimensional subspace
\begin{equation}\label{eq:WA_intro}
W_{\mathbf{A}} \triangleq \{(\bm{x}, \bm{x}\mathbf{A}) : \bm{x} \in \Fq^\mu\} \leq \Fqm.
\end{equation}
Given an offset vector $b(\mathbf{A})\in\Fq^{m-\mu}$, the kernel is defined as
\begin{equation}\label{eq:kernel_intro}
\Pi_A \triangleq W_\mathbf{A}+(\bm{0},b(\mathbf{A}))=\{(\bm{x},b(\mathbf{A})+\bm{x}\mathbf{A}):\bm{x}\in\Fq^\mu\}.
\end{equation}
Then, $\Pi_{\mathbf{A}}$ is a $\mu$-flat with direction $W_{\mathbf{A}}$ in $\AG(m,q)$, and the corresponding PCRB is extracted from the full pencil~$\cG(\Pi_{\mathbf{A}})$.

\subsubsection{The Case~$\mu\geq 2$} 
Let $n=\min(\mu,m-\mu)$ and $M=\max(\mu,m-\mu)$. Construct an $(n, n-1, 2)$ Gabidulin code over $\mathbb{F}_{q^M}$. If $\mu\leq m-\mu$, this code is directly used as a matrix set $\mathcal{A}\subseteq\Fq^{\mu\times(m-\mu)}$. If $\mu>m-\mu$, we transpose the resulting matrices to obtain $\mathcal{A}\subseteq\Fq^{\mu\times(m-\mu)}$. In both cases,
\begin{equation}
  \rank_{\Fq}(\mathbf{A}_i-\mathbf{A}_j)\ge 2, \quad \forall \mathbf{A}_i\ne \mathbf{A}_j\in\mathcal{A}.
\end{equation}
Consequently, by the lifting identity in~\eqref{eq:subspace_distance}, 
\begin{equation}
  d_{\rm S}(W_{\mathbf{A}_i},W_{\mathbf{A}_j})=d_{\rm S}(\mathcal{L}(\mathbf{A}_i),\mathcal{L}(\mathbf{A}_j))\ge 4,
\end{equation}
which is equivalent to $\dim(W_{\mathbf{A}_i}\cap W_{\mathbf{A}_j})\!\leq\! \mu\!-\!2$. By Theorem~\ref{lem:disjoint_cond}, the resulting pencils $\{\cG(\Pi_\mathbf{A})\}_{\mathbf{A}\in\mathcal{A}}$ are pairwise disjoint for any offset map $b(\mathbf{A})$. Hence, the lifted Gabidulin construction satisfies (P1) and offers $B=q^{(n-1)M}$ affine kernels for $\mu\geq 2$. The default offset choice is $b(\mathbf{A})=\bm{0}$, while non-zero offsets may be chosen to meet the rank-preservation condition (P2).

\subsubsection{The Case $\mu\!=\!1$}
\label{sec:mu1}
Setting $\mu=1$ leads to a trivial $(1,1,1)$ Gabidulin code over $\mathbb{F}_{q^{m-1}}$, which is insufficient for the lifting identity~\eqref{eq:subspace_distance} to achieve $d_{\rm{S}}\geq 4$. We therefore enforce (P1) by condition~\eqref{eq:ag_disjoint_offset} of Theorem~\ref{lem:disjoint_cond}. 
Identify $\mathbb{F}_q^m \cong \Fq\times\Fq^{m-1}$, fix an $\Fq$-linear isomorphism $\varphi:\Fq^{m-1}\rightarrow\mathbb{F}_{q^{m-1}}$, and choose $\lambda\in\mathbb{F}_{q^{m-1}}\setminus\Fq$. For brevity, we write $\bm{a}\lambda$ for $\varphi^{-1}(\varphi(\bm{a})\lambda)\in\mathbb{F}_q^{m-1}$. 
For each $\bm{a}\in\Fq^{m-1}$, define the line
\begin{equation}\label{eq:mu1_line}
\Pi_{\bm{a}}=\{\left(x,\;\bm{a}+x\bm{a}\lambda):x\in\Fq\right\},
\end{equation}
whose direction and offset are
\begin{equation}
W_{\bm{a}}\!=\!\{(x,x\bm{a}\lambda):x\in\Fq\}\;\;\mathrm{and}\;\;b(\bm{a})\!=\!\bm{a},
\end{equation}
respectively. For any $\bm{a}_i\!\neq\! \bm{a}_j$, the vectors $(1,\bm{a}_i\lambda)$ and $(1,\bm{a}_j\lambda)$ are linearly independent over $\Fq$, hence $\dim(W_{\bm{a}_i}\cap W_{\bm{a}_j})\!=\!0\!=\!\mu\!-\!1$. Suppose for contradiction that $(0,\bm{a}_i-\bm{a}_j)\in W_{\bm{a}_i}+W_{\bm{a}_j}$. Then, there exist $\alpha,\beta\in\mathbb{F}_q$ such that $ (0,\bm{a}_i-\bm{a}_j)=\alpha(1,\bm{a}_i\lambda)+\beta(1,\bm{a}_j\lambda)$,
which leads to $\beta\!=\!-\alpha$ and $\bm{a}_i\!-\!\bm{a}_j\!=\!\alpha(\bm{a}_i\!-\!\bm{a}_j)\lambda$. Since $\bm{a}_i\neq \bm{a}_j$ forces $\alpha\neq0$, this implies $\lambda\!=\!\alpha^{-1}\in\Fq$, contradicting $\lambda\in\mathbb{F}_{q^{m-1}}\setminus\mathbb{F}_q$. Therefore, by Theorem~\ref{lem:disjoint_cond}, the pencils $\{\cG(\Pi_{\bm{a}}):\bm{a}\in\mathbb{F}_q^{m-1}\}$ are pairwise disjoint, which provides $B=q^{m-1}$ affine kernels for $\mu\!=\!1$. 

\subsubsection{Rank Verification for Lengths up to $1024$}
\label{sec:rank_offset}
\begin{table}[t]
\centering
\caption{Lifted Gabidulin construction and matrix sparsification for affine geometry codes with lengths up to $1024$}
\label{tab:ag_catalog}
\scriptsize
\setlength{\tabcolsep}{2.2pt}
\renewcommand{\arraystretch}{0.87}
\vspace{-3mm}
\begin{tabular}{@{}cccccccccccc@{}}
\toprule 
$m$ & $s$ & $\mu$ & $N$ & $K$ & $\mathcal C_{\mathrm {Gab}}$ & $b$ & $J$ & $B$ & $\rho$ & $w_{\mathrm r}$ & $\tilde{w}_r$ \\
\midrule
$3$ & $1$ & $1$ & $8$ & $4$ & $(1,1,1)_{2^2}$ & $\bm{a}$  & $3$ & $4$ & $85.7\%$ & $4$ & $3$ \\
$4$ & $1$ & $2$ & $16$ & $11$ & $(2,1,2)_{2^2}$ & $0$ & $3$ & $4$ & $40\%$ & $8$ & $5$ \\
$4$ & $1$ & $1$ & $16$ & $5$ & $(1,1,1)_{2^3}$ & $\bm{a}$  & $7$ & $8$ & $40\%$ & $4$ & $3$ \\
$5$ & $1$ & $3$ & $32$ & $26$ & $(2,1,2)_{2^3}$ & $0$ & $3$ & $8$ & $38.7\%$ & $16$ & $9$ \\
$5$ & $1$ & $2$ & $32$ & $16$ & $(2,1,2)_{2^3}$ & $0$ & $7$ & $8$ & $9.03\%$ & $8$ & $5$ \\
$5$ & $1$ & $1$ & $32$ & $6$ & $(1,1,1)_{2^4}$ &  $\bm{a}$ & $15$ & $16$ & $19.4\%$ & $4$ & $3$ \\
$6$ & $1$ & $4$ & $64$ & $57$ & $(2,1,2)_{2^4}$ & $0$ & $3$ & $16$ & $38.1\%$ & $32$ & $17$ \\
$6$ & $1$ & $3$ & $64$ & $42$ & $(3,2,2)_{2^3}$ & $0$ & $7$ & $64$ & $17.2\%$ & $16$ & $9$ \\
$6$ & $1$ & $2$ & $64$ & $22$ & $(2,1,2)_{2^4}$ & $0$ & $15$ & $16$ & $2.15\%$ & $8$ & $5$ \\
$6$ & $1$ & $1$ & $64$ & $7$ & $(1,1,1)_{2^5}$ & $\bm{a}$  & $31$ & $32$ & $9.52\%$ & $4$ & $3$ \\
$3$ & $2$ & $1$ & $64$ & $48$ & $(1,1,1)_{4^2}$ & $\bm{a}$  & $5$ & $16$ & $95.2\%$ & $16$ & $\{13,5\}$ \\
$7$ & $1$ & $5$ & $128$ & $120$ & $(2,1,2)_{2^5}$ & $0$ & $3$ & $32$ & $37.8\%$ & $64$ & $33$ \\
$7$ & $1$ & $4$ & $128$ & $99$ & $(3,2,2)_{2^4}$ & $0$ & $7$ & $256$ & $16.8\%$ & $32$ & $17$ \\
$7$ & $1$ & $3$ & $128$ & $64$ & $(3,2,2)_{2^4}$ & $0$ & $15$ & $256$ & $4.06\%$ & $16$ & $9$ \\
$7$ & $1$ & $2$ & $128$ & $29$ & $(2,1,2)_{2^5}$ & $0$ & $31$ & $32$ & $0.52\%$ & $8$ & $5$ \\
$7$ & $1$ & $1$ & $128$ & $8$ & $(1,1,1)_{2^6}$ & $\bm{a}$  & $63$ & $64$ & $4.72\%$ & $4$ & $3$ \\
$8$ & $1$ & $6$ & $256$ & $247$ & $(2,1,2)_{2^6}$ & $0$ & $3$ & $64$ & $37.6\%$ & $128$ & $65$ \\
$8$ & $1$ & $5$ & $256$ & $219$ & $(3,2,2)_{2^5}$ & $0$ & $7$ & $1024$ & $16.6\%$ & $64$ & $33$ \\
$8$ & $1$ & $4$ & $256$ & $163$ & $(4,3,2)_{2^4}$ & $0$ & $15$ & $4096$ & $7.9\%$ & $32$ & $17$ \\
$8$ & $1$ & $3$ & $256$ & $93$ & $(3,2,2)_{2^5}$ & $0$ & $31$ & $1024$ & $0.99\%$ & $16$ & $9$ \\
$8$ & $1$ & $2$ & $256$ & $37$ & $(2,1,2)_{2^6}$ & $0$ & $63$ & $64$ & $0.13\%$ & $8$ & $5$ \\
$8$ & $1$ & $1$ & $256$ & $9$ & $(1,1,1)_{2^7}$ & $\bm{a}$  & $127$ & $128$ & $2.35\%$ & $4$ & $3$ \\
$4$ & $2$ & $2$ & $256$ & $231$ & $(2,1,2)_{4^2}$ & $b_2$ & $5$ & $16$ & $23.5\%$ & $64$ & $\{49,17\}$ \\
$4$ & $2$ & $1$ & $256$ & $127$ & $(1,1,1)_{4^3}$ & $\bm{a}$  & $21$ & $64$ & $23.5\%$ & $16$ & $\{13,5\}$ \\
$9$ & $1$ & $7$ & $512$ & $502$ & $(2,1,2)_{2^7}$ & $0$ & $3$ & $128$ & $37.6\%$ & $256$ & $129$ \\
$9$ & $1$ & $6$ & $512$ & $466$ & $(3,2,2)_{2^6}$ & $0$ & $7$ & $4096$ & $16.5\%$ & $128$ & $65$ \\
$9$ & $1$ & $5$ & $512$ & $382$ & $(4,3,2)_{2^5}$ & $0$ & $15$ & $32768$ & $7.8\%$ & $64$ & $33$ \\
$9$ & $1$ & $4$ & $512$ & $256$ & $(4,3,2)_{2^5}$ & $0$ & $31$ & $32768$ & $1.92\%$ & $32$ & $17$ \\
$9$ & $1$ & $3$ & $512$ & $130$ & $(3,2,2)_{2^6}$ & $0$ & $63$ & $4096$ & $0.24\%$ & $16$ & $9$ \\
$9$ & $1$ & $2$ & $512$ & $46$ & $(2,1,2)_{2^7}$ & $0$ & $127$ & $128$ & $0.03\%$ & $8$ & $5$ \\
$9$ & $1$ & $1$ & $512$ & $10$ & $(1,1,1)_{2^8}$ & $\bm{a}$  & $255$ & $256$ & $1.17\%$ & $4$ & $3$ \\
$3$ & $3$ & $1$ & $512$ & $448$ & $(1,1,1)_{8^2}$ & $\bm{a}$  & $9$ & $64$ & $98.6\%$ & $64$ & $\{57,9\}$ \\
$10$ & $1$ & $8$ & $1024$ & $1013$ & $(2,1,2)_{2^8}$ & $0$ & $3$ & $256$ & $37.5\%$ & $512$ & $257$ \\
$10$ & $1$ & $7$ & $1024$ & $968$ & $(3,2,2)_{2^7}$ & $0$ & $7$ & $16384$ & $16.5\%$ & $256$ & $129$ \\
$10$ & $1$ & $6$ & $1024$ & $848$ & $(4,3,2)_{2^6}$ & $0$ & $15$ & $262144$ & $7.74\%$ & $128$ & $65$ \\
$10$ & $1$ & $5$ & $1024$ & $638$ & $(5,4,2)_{2^5}$ & $0$ & $31$ & $1048576$ & $3.78\%$ & $64$ & $33$ \\
$10$ & $1$ & $4$ & $1024$ & $386$ & $(4,3,2)_{2^6}$ & $0$ & $63$ & $262144$ & $0.47\%$ & $32$ & $17$ \\
$10$ & $1$ & $3$ & $1024$ & $176$ & $(3,2,2)_{2^7}$ & $0$ & $127$ & $16384$ & $0.06\%$ & $16$ & $9$ \\
$10$ & $1$ & $2$ & $1024$ & $56$ & $(2,1,2)_{2^8}$ & $0$ & $255$ & $256$ & $0.008\%$ & $8$ & $5$ \\
$10$ & $1$ & $1$ & $1024$ & $11$ & $(1,1,1)_{2^9}$ & $\bm{a}$  & $511$ & $512$ & $0.59\%$ & $4$ & $3$ \\
$5$ & $2$ & $3$ & $1024$ & $988$ & $(2,1,2)_{4^3}$ & $b_1$ & $5$ & $64$ & $23.5\%$ & $256$ & $\{193,65\}$ \\
$5$ & $2$ & $2$ & $1024$ & $748$ & $(2,1,2)_{4^3}$ & $b_2$ & $21$ & $64$ & $1.45\%$ & $64$ & $\{49,17\}$ \\
$5$ & $2$ & $1$ & $1024$ & $288$ & $(1,1,1)_{4^4}$ & $\bm{a}$  & $85$ & $256$ & $5.87\%$ & $16$ & $\{13,5\}$ \\
\bottomrule
\end{tabular}
\vspace{0.35em}

\begin{minipage}{0.48\textwidth}
\scriptsize
$N$ and $K$ denote the blocklength and dimension of the FG code, $\rho$ is the fraction of edges in $\mathbf{H}$ retained by $\mathbf{H}_p$, and the row weights after sparsification are listed as~$\tilde{w}_r$. 

For offsets marked $b_1$ and $b_2$, let $\eta(\mathbf{A})\in\mathbb{F}_q^{m-\mu}$ be the first $m\!-\!\mu$ $\Fq$-coordinates of $u_0$ in a fixed $\Fq$-basis of $\mathbb{F}_{q^M}$, then $b_1(\mathbf{A})=\eta(\mathbf{A})$ and $b_{2,j}(\mathbf{A})=\eta_j(\mathbf{A})+\eta_{j+1}(\mathbf{A})\eta_{j+2}(\mathbf{A})$, $j=0,1,\ldots,m\!-\!\mu\!-\!1$ ($\mathrm{mod}\,(m\!-\!\mu)$).
\end{minipage}
\vspace{-2mm} 
\end{table} 

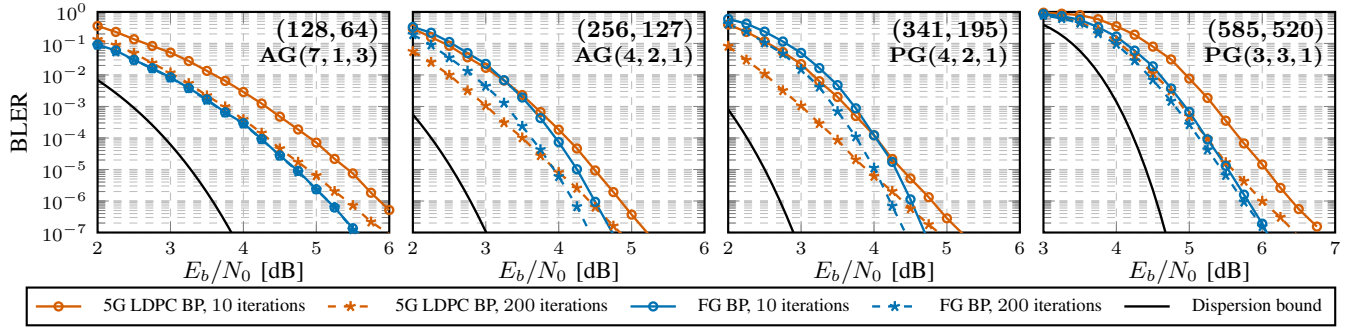
\begin{figure*}[t]
\centering
\input{figures/Fig1_128_70.tex}
\vspace{-3mm} 
\caption{BLER performance of flooding BP decoding on 5G LDPC matrices and ssPCMs of FG codes over binary additive white Gaussian noise channels.}
\vspace{-5mm}
\label{fig:fer}
\end{figure*}
The lifted Gabidulin construction ensures (P1), but not necessarily (P2). For the affine geometry, the offset map affects the binary row space of the selected pencils and must be chosen appropriately. Here, we verify $\rank_{\Ft}(\mathbf{H}_{p})=\rank_{\Ft}(\mathbf{H}_{\mathrm{AG}})$~(P2) directly for all affine geometry codes of lengths up to $1024$. 

When $q=2$, the zero offset $b(\mathbf{A})=\bm{0}$ satisfies the rank condition (P2) for all corresponding Reed-Muller codes. For $\mu=1$, the construction in~\eqref{eq:mu1_line} also satisfies (P2) for all tested cases. In contrast, for $q>2$, the zero-offset construction fails~(P2), as shown by Theorem~\ref{thm:zero_offset_rank_def}. We therefore search over a small family of offset maps $b:\mathcal{A}\rightarrow \mathbb{F}_q^{m-\mu}$, including linear and low-degree non-linear maps. Our empirical search shows that nonlinear offsets are often necessary for achieving the desired rank properties. The verified parameter sets are summarized in Table~\ref{tab:ag_catalog}.
\begin{theorem}\label{thm:zero_offset_rank_def}
Assume $q>2$ and $2\leq\mu\leq m-2$. Let $\mathbf{H}_{p}$ be the parity-check matrix obtained from $B$ pencils that are selected by the lifted Gabidulin construction with zero offset $b(\mathbf{A})=\bm{0}$. Then $\rank_{\Ft}(\mathbf{H}_{p})<\rank_{\Ft}(\mathbf{H}_\mathrm{AG})$.
\end{theorem}
\begin{proof}
Since $b(\mathbf{A})=\bm{0}$, every selected kernel contains the origin. Hence, every selected $(\mu+1)$-flat is a linear subspace of $\Fq^m$. Identifying each row~$v_F$ with indicator function $1_F:\mathbb{F}_q^m\rightarrow \mathbb{F}_2$, we have 
\begin{equation}\label{eq:scalar_invariance}
v_F(\bm{x})=v_F(\gamma\bm{x})
\end{equation}
for every selected row, every $\bm{x}\in\Fq^m$, and every $\gamma\in\Fq^\ast$. This property is preserved under binary linear combinations, so every vector in $\mathrm{row}_{\Ft}(\mathbf{H}_{p})$ satisfies~\eqref{eq:scalar_invariance}. 

We now exhibit a row~$v_{F^{\prime}}$ of $\mathbf{H}_{\mathrm{AG}}$ whose indicator function does not satisfy~\eqref{eq:scalar_invariance}. 
Given $2\!\leq\!\mu\!\le\! m-2$, choose a $(\mu+1)$-dimensional subspace $U\le\Fq^m$ and a point
$\bm{z}\notin U$. The affine flat $F'=\bm{z}+U$ is a valid row of $\mathbf{H}_\mathrm{AG}$. Choose any $\gamma\in\mathbb{F}_q^\ast\setminus\{1\}$, which is non-empty since $q>2$. Then $(\gamma-1)\bm{z}\notin U$, so $\gamma\bm{z}\notin \bm{z}+U=F^\prime$. Hence, $v_{F'}$ does not satisfy~\eqref{eq:scalar_invariance}.  
Since 
$\mathrm{row}_{\Ft}(\mathbf{H}_{p})\!\subseteq\!\mathrm{row}_{\Ft}(\mathbf{H}_\mathrm{AG})$, $v_{F'}\notin \mathrm{row}_{\Ft}(\mathbf{H}_p)$, and $v_{F'}\in \mathrm{row}_{\Ft}(\mathbf{H}_\mathrm{AG})$, we have $\rank_{\Ft}(\mathbf{H}_{p})<\rank_{\Ft}(\mathbf{H}_\mathrm{AG})$.
\end{proof}

\subsection{Pencil Selection for Projective Geometry Codes}\label{sec:pg_construction} 
\subsubsection{Gabidulin Construction}
Unlike in the affine case, no offset is involved, since the projective flats are represented by linear subspaces. Identify $\mathbb{F}_q^{m+1}\cong \Fq^{\mu+1}\times \Fq^{m-\mu}$. For a matrix $\mathbf{A}\in\mathbb{F}_q^{(\mu+1)\times(m-\mu)}$, define the $(\mu+1)$-dimensional subspace
\begin{equation}
  W_\mathbf{A}\triangleq\{(\bm{x},\bm{x}\mathbf{A}):\bm{x}\in\Fq^{\mu+1}\}\leq\mathbb{F}_q^{m+1}.
\end{equation}
The corresponding projective kernel is 
\begin{equation}
  \Pi_\mathbf{A}\triangleq\mathbb{P}(W_\mathbf{A}),
\end{equation}
which is a projective $\mu$-flat in $\PG(m,q)$. 

Let $n=\min(\mu+1,m-\mu)$ and $M=\max(\mu+1,m-\mu)$. Construct an $(n, n\!-\!1, 2)$ Gabidulin code over $\mathbb{F}_{q^M}$. After expanding its codewords over $\Fq$ and transposing if necessary, we obtain a matrix set $\mathcal{A}\subseteq \mathbb{F}_q^{(\mu+1)\times(m-\mu)}$. The lifting identity~\eqref{eq:subspace_distance} gives $d_{\rm S}(W_{\mathbf{A}_i},W_{\mathbf{A}_j})\geq 4$ for $\mathbf{A}_i,\mathbf{A}_j\in\mathcal{A}$, equivalently, $\dim(W_{\mathbf{A}_i}\cap W_{\mathbf{A}_j})\leq\mu\!-\!1$. Therefore, the full pencils $\{\cG(\Pi_\mathbf{A}):\mathbf{A}\in\mathcal{A}\}$ are pairwise disjoint, providing $B \!=\! q^{(n-1)M}$ projective kernels. 
\subsubsection{Rank Verification for Lengths up to $1024$} 
We perform the pencil selection for all projective geometry codes of lengths up to $1024$. As shown in Table~\ref{tab:pg_catalog}, the results for all tested cases achieve $\rank_{\Ft}(\mathbf{H}_p)=\rank_{\Ft}(\mathbf{H}_\mathrm{PG})$~(P2). 
\begin{table}[t] 
\centering
\caption{Lifted Gabidulin construction and matrix sparsification for projective geometry codes with lengths up to $1024$.}
\label{tab:pg_catalog}
\scriptsize
\setlength{\tabcolsep}{2.8pt}
\renewcommand{\arraystretch}{0.87}
\vspace{-2mm}
\begin{tabular}{@{}ccccccccccc@{}}
\toprule
$m$ & $s$ & $\mu$ & $N$ & $K$ & $\mathcal C_{\mathrm {Gab}}$ & $J$ & $B$ & $\rho$ & $w_{\mathrm r}$ & $\tilde{w}_r$ \\
\midrule
$3$ & $1$ & $1$ & $15$ & $10$ & $(2,1,2)_{2^2}$ & $3$ & $4$ & $80\%$ & $7$ & $\{5,4\}$ \\
$4$ & $1$ & $2$ & $31$ & $25$ & $(2,1,2)_{2^3}$ & $3$ & $8$ & $77.4\%$ & $15$ & $\{9,8\}$ \\
$4$ & $1$ & $1$ & $31$ & $15$ & $(2,1,2)_{2^3}$ & $7$ & $8$ & $36.1\%$ & $7$ & $\{5,4\}$ \\
$5$ & $1$ & $3$ & $63$ & $56$ & $(2,1,2)_{2^4}$ & $3$ & $16$ & $76.2\%$ & $31$ & $\{17,16\}$ \\
$5$ & $1$ & $2$ & $63$ & $41$ & $(3,2,2)_{2^3}$ & $7$ & $64$ & $68.8\%$ & $15$ & $\{9,8\}$ \\
$5$ & $1$ & $1$ & $63$ & $21$ & $(2,1,2)_{2^4}$ & $15$ & $16$ & $17.2\%$ & $7$ & $\{5,4\}$ \\
$3$ & $2$ & $1$ & $85$ & $68$ & $(2,1,2)_{4^2}$ & $5$ & $16$ & $94.1\%$ & $21$ & $\{17,6\}$ \\
$6$ & $1$ & $4$ & $127$ & $119$ & $(2,1,2)_{2^5}$ & $3$ & $32$ & $75.6\%$ & $63$ & $\{33,32\}$ \\
$6$ & $1$ & $3$ & $127$ & $98$ & $(3,2,2)_{2^4}$ & $7$ & $256$ & $67.2\%$ & $31$ & $\{17,16\}$ \\
$6$ & $1$ & $2$ & $127$ & $63$ & $(3,2,2)_{2^4}$ & $15$ & $256$ & $32.5\%$ & $15$ & $\{9,8\}$ \\
$6$ & $1$ & $1$ & $127$ & $28$ & $(2,1,2)_{2^5}$ & $31$ & $32$ & $8.4\%$ & $7$ & $\{5,4\}$ \\
$7$ & $1$ & $5$ & $255$ & $246$ & $(2,1,2)_{2^6}$ & $3$ & $64$ & $75.3\%$ & $127$ & $\{65,64\}$ \\
$7$ & $1$ & $4$ & $255$ & $218$ & $(3,2,2)_{2^5}$ & $7$ & $1024$ & $66.4\%$ & $63$ & $\{33,32\}$ \\
$7$ & $1$ & $3$ & $255$ & $162$ & $(4,3,2)_{2^4}$ & $15$ & $4096$ & $63.2\%$ & $31$ & $\{17,16\}$ \\
$7$ & $1$ & $2$ & $255$ & $92$ & $(3,2,2)_{2^5}$ & $31$ & $1024$ & $15.8\%$ & $15$ & $\{9,8\}$ \\
$7$ & $1$ & $1$ & $255$ & $36$ & $(2,1,2)_{2^6}$ & $63$ & $64$ & $4.15\%$ & $7$ & $\{5,4\}$ \\
$4$ & $2$ & $2$ & $341$ & $315$ & $(2,1,2)_{4^3}$ & $5$ & $64$ & $93.8\%$ & $85$ & $\{65,22\}$ \\
$4$ & $2$ & $1$ & $341$ & $195$ & $(2,1,2)_{4^3}$ & $21$ & $64$ & $23.2\%$ & $21$ & $\{17,6\}$ \\
$8$ & $1$ & $6$ & $511$ & $501$ & $(2,1,2)_{2^7}$ & $3$ & $128$ & $75.1\%$ & $255$ & $\{129,128\}$ \\
$8$ & $1$ & $5$ & $511$ & $465$ & $(3,2,2)_{2^6}$ & $7$ & $4096$ & $66\%$ & $127$ & $\{65,64\}$ \\
$8$ & $1$ & $4$ & $511$ & $381$ & $(4,3,2)_{2^5}$ & $15$ & $32768$ & $62.4\%$ & $63$ & $\{33,32\}$ \\
$8$ & $1$ & $3$ & $511$ & $255$ & $(4,3,2)_{2^5}$ & $31$ & $32768$ & $30.7\%$ & $31$ & $\{17,16\}$ \\
$8$ & $1$ & $2$ & $511$ & $129$ & $(3,2,2)_{2^6}$ & $63$ & $4096$ & $7.8\%$ & $15$ & $\{9,8\}$ \\
$8$ & $1$ & $1$ & $511$ & $45$ & $(2,1,2)_{2^7}$ & $127$ & $128$ & $2.1\%$ & $7$ & $\{5,4\}$ \\
$3$ & $3$ & $1$ & $585$ & $520$ & $(2,1,2)_{8^2}$ & $9$ & $64$ & $98.5\%$ & $73$ & $\{65,10\}$ \\
$9$ & $1$ & $7$ & $1023$ & $1012$ & $(2,1,2)_{2^8}$ & $3$ & $256$ & $75.1\%$ & $511$ & $\{257,256\}$ \\
$9$ & $1$ & $6$ & $1023$ & $967$ & $(3,2,2)_{2^7}$ & $7$ & $16384$ & $65.8\%$ & $255$ & $\{129,128\}$ \\
$9$ & $1$ & $5$ & $1023$ & $847$ & $(4,3,2)_{2^6}$ & $15$ & $262144$ & $61.9\%$ & $127$ & $\{65,64\}$ \\
$9$ & $1$ & $4$ & $1023$ & $637$ & $(5,4,2)_{2^5}$ & $31$ & $1048576$ & $60.5\%$ & $63$ & $\{33,32\}$ \\
$9$ & $1$ & $3$ & $1023$ & $385$ & $(4,3,2)_{2^6}$ & $63$ & $262144$ & $15.1\%$ & $31$ & $\{17,16\}$ \\
$9$ & $1$ & $2$ & $1023$ & $175$ & $(3,2,2)_{2^7}$ & $127$ & $16384$ & $3.87\%$ & $15$ & $\{9,8\}$ \\
$9$ & $1$ & $1$ & $1023$ & $55$ & $(2,1,2)_{2^8}$ & $255$ & $256$ & $1.03\%$ & $7$ & $\{5,4\}$ \\
\bottomrule
\end{tabular}
\vspace{-1mm}
\end{table}

\section{Numerical Results}\label{sec:sim}
As reported in Tables~\ref{tab:ag_catalog} and~\ref{tab:pg_catalog}, for both affine and projective geometry codes, the pencil-selected matrix~$\mathbf{H}_p$ retains only a fraction~$\rho$ of the edges in the full incidence matrix~$\mathbf{H}$, and the row weight can be further reduced from $w_r$ to $\tilde{w}_r$ after matrix sparsification. Fig.~\ref{fig:fer} presents the BLER performance of flooding BP on the ssPCM for four FG codes. The results show fast convergence, no visible error floor, and around $0.5$~dB gain over corresponding 5G LDPC codes at a BLER of $10^{-7}$. 

\section{Conclusion}
In this paper, we propose a geometric method to derive LDPC representations for FG codes. We solve the equivalent pencil selection problem using lifted Gabidulin codes, and verify rank preservation for FG codes up to length~$1024$ by computation. This work extends the availability of FG codes for BP decoding, which enables them to combine the algebraic distance properties with iterative soft-decision decoding.

\section*{Acknowledgment}
This work was supported by the Swiss State Secretariat for Education, Research, and Innovation under the SwissChips initiative. 
The authors thank Zongyao Li for helpful discussions.
\bibliographystyle{IEEEtran}
\bibliography{IEEEabrv,mybib} 

\end{document}

%% file: figures/Fig1_128_70.tex
\usepgfplotslibrary{groupplots}
\begin{tikzpicture}
\definecolor{myblued}{RGB}{0, 114, 178}
\definecolor{myred}{RGB}{213, 94, 0}
\definecolor{myredd}{RGB}{248,74,173}
\definecolor{myyellow}{RGB}{237,137,32}
\definecolor{mypurple}{RGB}{142,108,212}
\definecolor{myblues}{RGB}{77,190,238}
\definecolor{mygreen}{RGB}{0,160,135}
  \pgfplotsset{
    label style = {font=\fontsize{9pt}{7.2}\selectfont},
    tick label style = {font=\fontsize{7pt}{7.2}\selectfont}
  }

\usetikzlibrary{
    matrix,
}
\begin{axis}[
	scale = 1,
    ymode=log,
    xlabel={$E_b/N_0$ [\text{dB}]}, xlabel style={yshift=0.5em},
    ylabel={BLER}, ylabel style={yshift=-0.15em},
    title style={yshift=-7.5pt},
    grid=both,
    ymajorgrids=true,
    xmajorgrids=true,
    xmin=2,xmax=6,
    ymin=1E-07,ymax=1,
    ytick={0.00000001,0.0000001,0.000001,0.00001,0.0001,0.001,0.01,0.1,1},
    yticklabels={$10^{-8}$,$10^{-7}$,$10^{-6}$,$10^{-5}$,$10^{-4}$,$10^{-3}$,$10^{-2}$,$10^{-1}$,$10^{0}$},
    xtick={2,3,4,5,6,7,8},
    xticklabels={$2$,$3$,$4$,$5$,$6$,$7$,$8$},
    grid style=dashed,
    xshift=-1\textwidth,
    width=0.30\textwidth, height=4.5cm,
    thick,
    legend style={
      nodes={scale=1, transform shape},
      at={(2.00,-0.335)},
      anchor={center},
      cells={anchor=west},
      column sep= 2.1mm,
      row sep= -1mm,
      font=\fontsize{6.9pt}{7.2}\selectfont,
    },
    legend columns=5,
    ]

\addplot[
    color=myred,
    mark=*,
    fill opacity=0,
    line width=0.31mm,
    mark size=1.5,
]
table {
2.000000    3.5842e-01
2.250000    2.3155e-01
2.500000    1.3587e-01
2.750000    8.2988e-02
3.000000    5.1329e-02
3.250000    2.7647e-02
3.500000    1.3463e-02
3.750000    6.4792e-03
4.000000    2.8709e-03
4.250000    1.2387e-03
4.500000    4.7451e-04
4.750000    1.8623e-04
5.000000    7.1478e-05
5.250000    2.1687e-05
5.500000    7.5234e-06
5.750000    1.8268e-06
6.000000    5.1767e-07
};
\addlegendentry{5G LDPC BP, $10$~iterations}

\addplot[
    dashed,
    color=myred,
    mark=star,
    line width=0.31mm,
    mark size=2,
]
table {
2.00    1.3300e-01
2.25    8.8574e-02
2.50    4.8921e-02
2.75    2.3842e-02
3.00    1.1643e-02
3.25    5.3772e-03
3.50    2.1971e-03
3.75    9.8140e-04
4.00    3.8558e-04
4.25    1.4339e-04
4.50    4.5595e-05
4.75    1.7156e-05
5.00    6.3701e-06
5.25    2.0336e-06
5.50    7.2157e-07
5.75    2.1766e-07
6.00    8.9194e-08
};
\addlegendentry{5G LDPC BP, $200$~iterations}

\addplot[
    color=myblued,
    mark=*,
    fill opacity=0,
    line width=0.31mm,
    mark size=1.5,
]
table {
2.000000    9.165903E-02
2.250000    5.966587E-02
2.500000    2.998051E-02
2.750000    1.635189E-02
3.000000    8.386096E-03
3.250000    3.853045E-03
3.500000    1.673318E-03
3.750000    6.452737E-04
4.000000    2.939067E-04
4.250000    9.257802E-05
4.500000    2.759889E-05
4.750000    9.035615E-06
5.000000    2.353153E-06
5.250000    6.217863E-07
5.500000    1.368670E-07
5.600000    8.633641E-08
};
\addlegendentry{FG BP, $10$~iterations}

\addplot[
    dashed,
    color=myblued,
    mark=star,
    fill opacity=0,
    line width=0.31mm,
    mark size=2,
]
table {
2.000000    8.857396E-02
2.250000    5.431831E-02
2.500000    2.955919E-02
2.750000    1.589067E-02
3.000000    8.067444E-03
3.250000    3.717610E-03
3.500000    1.570956E-03
3.750000    6.336030E-04
4.000000    2.797982E-04
4.250000    9.006810E-05
4.500000    2.715601E-05
4.750000    8.382521E-06
5.000000    2.328026E-06
5.250000    6.373355E-07
5.500000    1.243177E-07
};
\addlegendentry{FG BP, $200$~iterations}

\addplot[
    color=black,
    line width=0.31mm,
]
table {
1	0.102689677840519
1.10000000000000	0.0838004132460626
1.20000000000000	0.0674936370913222
1.30000000000000	0.0536193244083882
1.40000000000000	0.0419907749304545
1.50000000000000	0.0323950953629262
1.60000000000000	0.0246037667904041
1.70000000000000	0.0183826934487477
1.80000000000000	0.0135012151309750
1.90000000000000	0.00973968626279041
2	0.00689536709410712
2.10000000000000	0.00478652182275886
2.20000000000000	0.00325476025217670
2.30000000000000	0.00216578107488117
2.40000000000000	0.00140876659330466
2.50000000000000	0.000894735320003512
2.60000000000000	0.000554179591827080
2.70000000000000	0.000334303391183011
2.80000000000000	0.000196137581152757
2.90000000000000	0.000111754354312665
3	6.17391099148114e-05
3.10000000000000	3.30147649975997e-05
3.20000000000000	1.70574641320786e-05
3.30000000000000	8.49821713759520e-06
3.40000000000000	4.07412320353951e-06
3.50000000000000	1.87521478750640e-06
3.60000000000000	8.26649807726359e-07
3.70000000000000	3.48103129138328e-07
3.80000000000000	1.39631641449300e-07
3.90000000000000	5.31895882940531e-08
4	1.91782182808796e-08
4.10000000000000	6.52207133063432e-09
4.20000000000000	2.08395958318834e-09
4.30000000000000	6.23034768496742e-10
4.40000000000000	1.73498604489000e-10
};
\addlegendentry{Dispersion bound}

\node[right, align=left]
at (axis cs:4.3,2.51E-01) {\small{$\bm{(128,64)}$}};

\node[right, align=left]
at (axis cs:4.07,3.91E-02) {\footnotesize{$\bm{\mathrm{AG}(7,1,3)}$}};

\end{axis}

\begin{axis}[
	scale = 1,
    ymode=log,
    xlabel={$E_b/N_0$ [\text{dB}]}, xlabel style={yshift=0.5em},
    title style={yshift=-7.5pt},
    grid=both,
    ymajorgrids=true,
    xmajorgrids=true,
    xmin=2,xmax=6,
    ymin=1E-07,ymax=1,
    ytick={0.00000001,0.0000001,0.000001,0.00001,0.0001,0.001,0.01,0.1,1},
    yticklabels={},
    xtick={2,3,4,5,6,7,8},
    xticklabels={$2$,$3$,$4$,$5$,$6$,$7$,$8$},
    grid style=dashed,
    xshift=-0.77\textwidth,
    width=0.30\textwidth, height=4.5cm,
    thick,
    legend style={
      nodes={scale=1, transform shape},
      at={(0.50,-0.33)},
      anchor={center},
      cells={anchor=west},
      column sep= 1.6mm,
      row sep= -1mm,
      font=\fontsize{6.9pt}{7.2}\selectfont,
    },
    legend columns=2,
    ]

\addplot[
    color=myred,
    mark=*,
    fill opacity=0,
    line width=0.31mm,
    mark size=1.5,
]
table { 
2.000000    2.9326e-01
2.250000    1.5823e-01
2.500000    8.4388e-02
2.750000    3.6311e-02
3.000000    1.6844e-02
3.250000    6.6525e-03
3.500000    2.2990e-03
3.750000    6.6499e-04
4.000000    1.8371e-04
4.250000    4.6333e-05
4.500000    9.3050e-06
4.750000    1.9187e-06
5.000000    3.7022e-07
5.250000    7.8464e-08
};

\addplot[
    dashed,
    color=myred,
    mark=star,
    fill opacity=0,
    line width=0.31mm,
    mark size=2.0,
]
table {
2.00    5.5453e-02
2.25    2.5246e-02
2.50    9.8990e-03
2.75    3.1791e-03
3.00    1.0686e-03
3.25    3.0700e-04
3.50    1.0154e-04
3.75    2.7584e-05
4.00    8.0692e-06
4.25    2.5605e-06
4.50    6.4104e-07
4.75    1.6113e-07
5.00    4.3233e-08
};

\addplot[
    color=myblued,
    mark=*,
    fill opacity=0,
    line width=0.31mm,
    mark size=1.5,
]
table {
2.000000    3.327787E-01
2.250000    2.178649E-01
2.500000    1.112966E-01
2.750000    4.940371E-02
3.000000    2.243292E-02
3.250000    6.844159E-03
3.500000    1.925224E-03
3.750000    4.280758E-04
4.000000    7.411657E-05
4.250000    1.017890E-05
4.500000    9.428531E-07
4.75000     8.665534E-08
};

\addplot[
    dashed,
    color=myblued,
    mark=star,
    line width=0.31mm,
    mark size=2.0,
]
table {
2.000000    1.912046E-01
2.250000    8.988764E-02
2.500000    3.513086E-02
2.750000    1.323627E-02
3.000000    4.380489E-03
3.250000    1.292048E-03
3.500000    2.369371E-04
3.750000    4.279693E-05
4.000000    5.974677E-06
4.250000    6.659217E-07
4.500000    4.597786E-08
};

\addplot[
    color=black,
    line width=0.31mm,
]
table {
1	0.0568876153945202
1.10000000000000	0.0408242644259592
1.20000000000000	0.0285488085615721
1.30000000000000	0.0194324738223476
1.40000000000000	0.0128593851299216
1.50000000000000	0.00826272493740595
1.60000000000000	0.00514838613329766
1.70000000000000	0.00310648592116171
1.80000000000000	0.00181256234717881
1.90000000000000	0.00102113383659987
2	0.000554551916899665
2.10000000000000	0.000289824218887940
2.20000000000000	0.000145505337852998
2.30000000000000	7.00399354128812e-05
2.40000000000000	3.22592024702690e-05
2.50000000000000	1.41861449011927e-05
2.60000000000000	5.94262453276114e-06
2.70000000000000	2.36553205752671e-06
2.80000000000000	8.92436145177337e-07
2.90000000000000	3.18204147574173e-07
3	1.06908519384211e-07
3.10000000000000	3.37366193494716e-08
3.20000000000000	9.96502018252322e-09
3.30000000000000	2.74497096855909e-09
3.40000000000000	7.02356105276124e-10
3.50000000000000	1.66220706543072e-10
3.60000000000000	3.62183754601958e-11
};

\node[right, align=left]
at (axis cs:4.1,2.51E-01) {\small{$\bm{(256,127)}$}};

\node[right, align=left]
at (axis cs:4.07,3.91E-02) {\footnotesize{$\bm{\mathrm{AG}(4,2,1)}$}};

\end{axis}

\begin{axis}[
	scale = 1,
    ymode=log,
    xlabel={$E_b/N_0$ [\text{dB}]}, xlabel style={yshift=0.5em},
    title style={yshift=-7.5pt},
    grid=both,
    ymajorgrids=true,
    xmajorgrids=true,
    xmin=2,xmax=6,
    ymin=1E-07,ymax=1,
    ytick={0.00000001,0.0000001,0.000001,0.00001,0.0001,0.001,0.01,0.1,1},
    yticklabels={},
    xtick={2,3,4,5,6,7,8},
    xticklabels={$2$,$3$,$4$,$5$,$6$,$7$,$8$},
    grid style=dashed,
    xshift=-0.54\textwidth,
    width=0.30\textwidth, height=4.5cm,
    thick,
    legend style={
      nodes={scale=1, transform shape},
      at={(0.50,-0.33)},
      anchor={center},
      cells={anchor=west},
      column sep= 1.6mm,
      row sep= -1mm,
      font=\fontsize{6.9pt}{7.2}\selectfont,
    },
    legend columns=2,
    ]

\addplot[
    color=myred,
    mark=*,
    fill opacity=0,
    line width=0.31mm,
    mark size=1.5,
]
table {
2.000000    3.9683e-01
2.250000    2.3248e-01
2.500000    1.1669e-01
2.750000    5.3390e-02
3.000000    2.2060e-02
3.250000    6.3239e-03
3.500000    2.0009e-03 
3.750000    4.9836e-04 
4.000000    1.2233e-04
4.250000    2.1509e-05
4.500000    5.2043e-06 
4.750000    1.3195e-06
5.000000    2.8151e-07
5.250000    7.6069e-08
};

\addplot[
    dashed, 
    color=myred,
    mark=star,
    line width=0.31mm, 
    mark size=2,
]
table {
2.00    8.1900e-02
2.25    2.9824e-02
2.50    1.0861e-02
2.75    3.1972e-03
3.00    1.0495e-03
3.25    2.8667e-04
3.50    8.6065e-05
3.75    2.0644e-05
4.00    6.1185e-06
4.25    2.0322e-06
4.50    5.7387e-07
4.75    1.7957e-07
4.90    9.3618e-08
};

\addplot[
    color=myblued,
    mark=*,
    fill opacity=0,
    line width=0.31mm,
    mark size=1.5,
]
table {
2.000000    5.988024E-01
2.250000    4.237288E-01
2.500000    2.370694E-01
2.750000    1.150086E-01
3.000000    5.005005E-02
3.250000    1.650710E-02
3.500000    4.750594E-03
3.750000    8.736867E-04
4.000000    1.218252E-04
4.250000    1.742621E-05
4.500000    1.125779E-06
4.750000    5.105560E-08
};

\addplot[
    dashed,
    color=myblued,
    mark=star,
    fill opacity=0,
    line width=0.31mm,
    mark size=2.0,
]
table {
2.000000    4.140787E-01
2.250000    2.383790E-01
2.500000    1.048768E-01
2.750000    4.737091E-02
3.000000    1.502065E-02
3.250000    4.117683E-03
3.500000    6.987192E-04
3.750000    1.093714E-04
4.000000    1.120145E-05
4.250000    6.920243E-07
4.500000    4.539559E-08
};

\addplot[
    color=black,
    line width=0.31mm,
]
table {
1	0.124564369690168
1.10000000000000	0.0896406704160200
1.20000000000000	0.0623443182193662
1.30000000000000	0.0418317234776516
1.40000000000000	0.0270302309541796
1.50000000000000	0.0167889321145683
1.60000000000000	0.0100043925364167
1.70000000000000	0.00570798877977754
1.80000000000000	0.00311164820615339
1.90000000000000	0.00161717605244398
2	0.000799422381926198
2.10000000000000	0.000374957882601118
2.20000000000000	0.000166435778708213
2.30000000000000	6.97218430547835e-05
2.40000000000000	2.74834266024741e-05
2.50000000000000	1.01622280021145e-05
2.60000000000000	3.51289600037160e-06
2.70000000000000	1.13120010459262e-06
2.80000000000000	3.38016941847370e-07
2.90000000000000	9.33397591902390e-08
3	2.37134185188029e-08
3.10000000000000	5.51625632390289e-09
3.20000000000000	1.16890920618835e-09
3.30000000000000	2.24383712824010e-10
3.40000000000000	3.87857924938647e-11
3.50000000000000	5.99807283044217e-12
3.60000000000000	8.24073366480392e-13
};

\node[right, align=left]
at (axis cs:4.1,2.51E-01) {\small{$\bm{(341, 195)}$}};

\node[right, align=left]
at (axis cs:4.07,3.91E-02) {\footnotesize{$\bm{\mathrm{PG}(4,2,1)}$}};

\end{axis}

\begin{axis}[
	scale = 1,
    ymode=log,
    xlabel={$E_b/N_0$ [\text{dB}]}, xlabel style={yshift=0.5em},
    title style={yshift=-7.5pt},
    grid=both, 
    ymajorgrids=true,
    xmajorgrids=true,
    xmin=3,xmax=7,
    ymin=1E-07,ymax=1,
    ytick={0.00000001,0.0000001,0.000001,0.00001,0.0001,0.001,0.01,0.1,1},
    yticklabels={},
    xtick={3,4,5,6,7,8},
    xticklabels={$3$,$4$,$5$,$6$,$7$,$8$},
    grid style=dashed,
    xshift=-0.31\textwidth,
    width=0.30\textwidth, height=4.5cm,
    thick,
    legend style={
      nodes={scale=1, transform shape},
      at={(0.50,-0.25)},
      anchor={center},
      cells={anchor=west},
      column sep= 1.6mm,
      row sep= -1mm,
      font=\fontsize{6.9pt}{7.2}\selectfont,
    },
    legend columns=2,
    ]

\addplot[
    color=myred,
    mark=*,
    fill opacity=0,
    line width=0.31mm,
    mark size=1.5,
]
table {
2.00    1.0000e+00
2.25    1.0000e+00
2.50    1.0000e+00
2.75    9.9010e-01
3.00    9.4340e-01
3.25    8.7719e-01
3.50    7.9365e-01
3.75    5.8824e-01
4.00    3.5088e-01
4.25    1.8939e-01
4.50    7.9220e-02
4.75    3.1047e-02
5.00    7.6834e-03
5.25    1.8675e-03
5.50    3.5289e-04
5.75    6.7177e-05
6.00    1.4414e-05
6.25    2.6118e-06 
6.50    5.5569e-07
6.75    1.5459e-07
};

\addplot[
    dashed,
    color=myred,
    mark=star,
    fill opacity=0,
    line width=0.31mm,
    mark size=2,
]
table {
2.00    1.0000e+00
2.25    1.0000e+00
2.50    1.0000e+00
2.75    9.4340e-01
3.00    8.6207e-01
3.25    6.9444e-01
3.50    4.6729e-01
3.75    2.6596e-01
4.00    1.1919e-01
4.25    3.9841e-02
4.50    1.1232e-02
4.75    3.1419e-03
5.00    4.9030e-04
5.25    9.3129e-05
5.50    1.6686e-05
5.75    4.2443e-06
6.00    9.8508e-07
6.25    3.0544e-07
6.50    7.3566e-08
};

\addplot[
    color=myblued,
    mark=*,
    fill opacity=0,
    line width=0.31mm,
    mark size=1.5,
]
table {
2.000000    1.000000E+00
2.250000    1.000000E+00
2.500000    9.950249E-01
2.750000    9.389671E-01
3.000000    8.620690E-01
3.250000    7.326007E-01
3.500000    5.665103E-01
3.750000    3.264557E-01
4.000000    1.651528E-01
4.250000    5.913661E-02
4.500000    1.877758E-02
4.750000    3.730021E-03
5.000000    6.734846E-04
5.250000    9.033867E-05
5.500000    1.310781E-05
5.750000    1.586909E-06
6.000000    1.909826E-07
6.100000    6.748480E-08
};

\addplot[
    dashed,
    color=myblued,
    mark=star,
    fill opacity=0,
    line width=0.31mm,
    mark size=2,
]
table {
2.000000    1.000000E+00
2.250000    9.950249E-01
2.500000    9.803922E-01
2.750000    9.049774E-01
3.000000    7.812500E-01
3.250000    6.172840E-01
3.500000    4.405286E-01
3.750000    2.063983E-01
4.000000    9.433962E-02
4.250000    2.848191E-02
4.500000    6.946133E-03
4.750000    1.470653E-03
5.000000    2.660519E-04
5.250000    4.209365E-05
5.500000    6.532639E-06
5.750000    9.311812E-07
6.000000    1.553956E-07
};

\addplot[
    color=black,
    line width=0.31mm,
]
table {
1	0.999989183540688
1.10000000000000	0.999972504809139
1.20000000000000	0.999932912706821
1.30000000000000	0.999842864473086
1.40000000000000	0.999646665227469
1.50000000000000	0.999237231270701
1.60000000000000	0.998419112660029
1.70000000000000	0.996854356844155
1.80000000000000	0.993990871850521
1.90000000000000	0.988979611815858
2	0.980597549797872
2.10000000000000	0.967206513963807
2.20000000000000	0.946789140498408
2.30000000000000	0.917105140294422
2.40000000000000	0.875995559172770
2.50000000000000	0.821824836797236
2.60000000000000	0.753994217201590
2.70000000000000	0.673401746048587
2.80000000000000	0.582690602661830
2.90000000000000	0.486145985445597
3	0.389183180959476
3.10000000000000	0.297498934085140
3.20000000000000	0.216088087878746
3.30000000000000	0.148398306451768
3.40000000000000	0.0958658544919880
3.50000000000000	0.0579478907226179
3.60000000000000	0.0325940220864220
3.70000000000000	0.0169591765792077
3.80000000000000	0.00811134720464714
3.90000000000000	0.00354187383639782
4	0.00140148487483864
4.10000000000000	0.000498435349997580
4.20000000000000	0.000157899136575210
4.30000000000000	4.41129038740556e-05
4.40000000000000	1.07484771624745e-05
4.50000000000000	2.25605124205736e-06
4.60000000000000	4.02310624389936e-07
4.70000000000000	6.00132630081265e-08
4.80000000000000	7.35930945565442e-09
4.90000000000000	7.27461384429167e-10
};

\node[right, align=left]
at (axis cs:5.1,2.51E-01) {\small{$\bm{(585,520)}$}};

\node[right, align=left]
at (axis cs:5.07,3.91E-02) {\footnotesize{$\bm{\mathrm{PG}(3,3,1)}$}};

\end{axis}

\end{tikzpicture}

%% file: IEEEabrv.bib
@STRING{IEEE_J_COML       = "{IEEE} Commun. Lett."}

@STRING{IEEE_J_JSAC       = "{IEEE} J. Sel. Areas Commun."}

@STRING{IEEE_J_COM        = "{IEEE} Trans. Commun."}

@STRING{IEEE_J_IT         = "{IEEE} Trans. Inf. Theory"}

@STRING{IEEE_J_PROC       = "Proc. {IEEE}"}

@STRING{IEEE_M_COM        = "{IEEE} Commun. Mag."}


%% file: mybib.bib
@book{halsted1896synthetic,
  author    = {Halsted, George Bruce},
  title     = {Synthetic Projective Geometry},
  series    = {Mathematical Monographs},
  number    = {2},
  edition   = {4},
  publisher = {John Wiley \& Sons},
  address   = {New York, NY, USA},
  year      = {1906}
}

@article{huang2009two,
  title={Two reliability-based iterative majority-logic decoding algorithms for {LDPC} codes},
  author={Huang, Qin and Kang, Jingyu and Zhang, Li and Lin, Shu and Abdel-Ghaffar, Khaled},
  journal=IEEE_J_COM,
  volume={57},
  number={12},
  pages={3597--3606},
  year={2009},
  month={Dec.},
  publisher={IEEE}
}

@inproceedings{chen2018finite,
  title={Finite Hyperplane Codes: Minimum Distance and Majority-Logic Decoding},
  author={Chen, Chao and Liu, Haiyang and Bai, Baoming},
  booktitle={Proc. IEEE Int. Symp. Inf. Theory (ISIT)},
  pages={2500--2504},
  year={2018},
  organization={IEEE}
}

@article{liu2012smallest,
  title={On the smallest absorbing sets of {LDPC} codes from finite planes},
  author={Liu, Haiyang and Li, Yan and Ma, Lianrong and Chen, Jie},
  journal=IEEE_J_IT,
  volume={58},
  number={6},
  pages={4014--4020},
  year={2012},
  month={Jun.},
  publisher={IEEE}
}

@article{smarandache2007pseudo,
  title={Pseudo-Codeword Analysis of {Tanner} Graphs From Projective and {Euclidean} Planes},
  author={Smarandache, Roxana and Vontobel, Pascal O.},
  journal=IEEE_J_IT,
  volume={53},
  number={7},
  pages={2376--2393},
  year={2007},
  month={Jul.},
  doi={10.1109/TIT.2007.899563},
  publisher={IEEE}
}

@article{kamiya2007high,
  title={High-rate quasi-cyclic low-density parity-check codes derived from finite affine planes},
  author={Kamiya, Norifumi},
  journal=IEEE_J_IT,
  volume={53},
  number={4},
  pages={1444--1459},
  year={2007},
  month={Apr.},
  publisher={IEEE}
}

@article{kamiya2006quasi,
  title={Quasi-Cyclic Codes from a Finite Affine Plane},
  author={Kamiya, Norifumi and Fossorier, Marc P. C.},
  journal={Des. Codes Cryptogr.},
  volume={38},
  number={3},
  pages={311--329},
  year={2006},
  month={Mar.},
  publisher={Springer}
}

@article{tang2005codes,
  title={Codes on finite geometries},
  author={Tang, Heng and Xu, Jun and Lin, Shu and Abdel-Ghaffar, Khaled AS},
  journal=IEEE_J_IT,
  volume={51},
  number={2},
  pages={572--596},
  year={2005},
  month={Feb.},
  publisher={IEEE}
}

@article{kou2002low,
  title={Low-density parity-check codes based on finite geometries: a rediscovery and new results},
  author={Kou, Yu and Lin, Shu and Fossorier, Marc P. C.},
  journal=IEEE_J_IT,
  volume={47},
  number={7},
  pages={2711--2736},
  month={Nov.},
  year={2001},
  publisher={IEEE}
}

@article{lin1973multifold,
  title={Multifold {E}uclidean geometry codes},
  author={Lin, Shu},
  journal=IEEE_J_IT,
  volume={19},
  number={4},
  pages={537--548},
  year={1973},
  month={Jul.},
  publisher={IEEE}
}

@mastersthesis{rudolph1964,
  author={Rudolph, L. D.},
  title={Geometric Configuration and Majority Logic Decoding Codes},
  school={Univ. Oklahoma},
  address={Norman, OK, USA},
  year={1964}
}

@techreport{smith1967,
  author={Smith, K. J. C.},
  title={Majority Decodable Codes Derived from Finite Geometries},
  institution={Univ. North Carolina},
  address={Chapel Hill, NC, USA},
  number={561},
  type={Institute of Statistics Mimeo Series},
  year={1967}
}

@incollection{weldon1967euclidean,
  author={Weldon Jr., E. J.},
  title={Euclidean Geometry Cyclic Codes},
  booktitle={Combinatorial Mathematics and its Applications},
  editor={Bose, R. C. and Dowling, T. A.},
  publisher={Univ. North Carolina Press},
  address={Chapel Hill, NC, USA},
  pages={377--388},
  year={1967}
}

@article{delsarte1969geometric,
  author={Delsarte, Philippe},
  title={A Geometric Approach to a Class of Cyclic Codes},
  journal={J. Combin. Theory},
  volume={6},
  number={4},
  pages={340--358},
  month={May},
  year={1969},
  doi={10.1016/S0021-9800(69)80030-X}
}

@techreport{hartmann1972generalized,
  author={Hartmann, Carlos R. P. and Rudolph, Luther D.},
  title={Generalized Finite-Geometry Codes},
  institution={Syracuse Univ.},
  type={Electrical Engineering and Computer Science Technical Report},
  number={32},
  address={Syracuse, NY, USA},
  month={Apr.},
  year={1972}
}

@article{hartmann1974structure,
  author={Hartmann, Carlos R. P. and Ducey, James B. and Rudolph, Luther D.},
  title={On the Structure of Generalized Finite-Geometry Codes},
  journal=IEEE_J_IT,
  volume={20},
  number={2},
  pages={240--252},
  month={Mar.},
  year={1974},
  publisher={IEEE}
}

@inproceedings{shen2025belief,
  title={Belief propagation decoding for short codes on structured sparse parity-check matrices},
  author={Shen, Yifei and Li, Zongyao and Boutillon, Emmanuel and Song, Wenqing and Ren, Yuqing and Zhang, Chuan and You, Xiaohu and Burg, Andreas},
  booktitle={Proc IEEE Int. Symp. Inf. Theory (ISIT)},
  pages={1--6},
  year={2025}
}

@inproceedings{mackay2000relationships,
  author    = {David J. C. MacKay},
  title     = {Relationships between Sparse Graph Codes},
  booktitle = {Proc. Workshop Inf.-Based Induction Sci. (IBIS 2000)},
  month     = jul,
  year      = {2000},
  pages     = {257--270}
}

@article{dolecek2009analysis,
  title={Analysis of absorbing sets and fully absorbing sets of array-based {LDPC} codes},
  author={Dolecek, Lara and Zhang, Zhengya and Anantharam, Venkat and Wainwright, Martin J and Nikolic, Borivoje},
  journal=IEEE_J_IT,
  volume={56},
  number={1},
  pages={181--201},
  year={2009},
  month={Dec.},
  publisher={IEEE}
}

@article{shirvanimoghaddam2018short,
  title={Short block-length codes for ultra-reliable low latency communications},
  author={Shirvanimoghaddam, Mahyar and Mohammadi, Mohammad Sadegh and Abbas, Rana and Minja, Aleksandar and Yue, Chentao and Matuz, Balazs and Han, Guojun and Lin, Zihuai and Liu, Wanchun and Li, Yonghui and Johnson, Sarah and Vucetic, Branka},
  journal=IEEE_M_COM,
  volume={57},
  number={2},
  pages={130--137},
  year={Feb. 2019},
  publisher={IEEE}
}

@article{sankaranarayanan2005iterative,
  title={Iterative decoding of linear block codes: A parity-check orthogonalization approach},
  author={Sankaranarayanan, Sundararajan and Vasi{\'c}, Bane},
  journal=IEEE_J_IT,
  volume={51},
  number={9},
  pages={3347--3353},
  year={Sept. 2005},
  publisher={IEEE}
}

@inproceedings{yedidia2002generating,
  title={Generating code representations suitable for belief propagation decoding},
  author={Yedidia, Jonathan S and Chen, Jinghu and Fossorier, Marc P. C.},
  booktitle={Proc. Annual Allerton Conf. Commun. Control Comput.},
  volume={40},
  number={1},
  pages={447--456},
  year={2002}
}

@article{kumar2005graphical,
  title={On graphical representations of algebraic codes suitable for iterative decoding},
  author={Kumar, Vidya and Milenkovic, Olgica},
  journal=IEEE_J_COML,
  volume={9},
  number={8},
  pages={729--731},
  year={Aug. 2005},
  publisher={IEEE}
}

@ARTICLE{miao2024trends,
  author={Miao, Sisi and Kestel, Claus and Johannsen, Lucas and Geiselhart, Marvin and Schmalen, Laurent and Balatsoukas-Stimming, Alexios and Liva, Gianluigi and Wehn, Norbert and ten Brink, Stephan},
  journal=IEEE_J_PROC, 
  title={Trends in Channel Coding for {6G}}, 
  year={Jul. 2024},
  volume={112},
  number={7},
  pages={653-675}
}

@article{shen2025toward,
  title={Toward Universal Belief Propagation Decoding for Short Binary Block Codes},
  author={Shen, Yifei and Li, Zongyao and Ren, Yuqing and Boutillon, Emmanuel and Balatsoukas-Stimming, Alexios and Zhang, Chuan and You, Xiaohu and Burg, Andreas},
  journal=IEEE_J_JSAC,
  volume={43},
  number={4},
  pages={1135-1152},
  month={Feb.},
  year={2025}
}

@article{Silva2008,
  author={Silva, Danilo and Kschischang, Frank R. and K{\"o}tter, Ralf},
  title={A rank-metric approach to error control in random network coding},
  journal=IEEE_J_IT,
  volume={54},
  number={9},
  pages={3951--3967},
  month={Sep.},
  year={2008},
  doi={10.1109/TIT.2008.928291}
}

@article{Gabidulin1985,
  author={Gabidulin, Ernst M.},
  title={Theory of codes with maximum rank distance},
  journal={Probl. Inf. Transm.},
  volume={21},
  number={1},
  pages={1--12},
  year={1985}
}

@book{LinCostello2004,
  author={Lin, Shu and Costello, Daniel J.},
  title={Error Control Coding},
  edition={2},
  publisher={Prentice Hall},
  address={Upper Saddle River, NJ, USA},
  year={2004}
}

@article{delsarte1970generalized,
  title={On generalized {Reed-Muller} codes and their relatives},
  author={Delsarte, Philippe and Goethals, Jean-Marie and MacWilliams, F. Jessie},
  journal={Inf. Control},
  volume={16},
  number={5},
  pages={403--442},
  year={1970},
  month={Jul.},
  publisher={Elsevier}
}

@article{yu2025ordered,
  author={Yu, Shuyan and Zhang, Bin and Huang, Qin},
  title={Ordered Statistics Derivative Decoding for Affine-Invariant Codes Without {Gaussian} Elimination},
  journal=IEEE_J_COM,
  volume={73},
  number={11},
  pages={11006--11016},
  month={Nov.},
  year={2025},
  doi={10.1109/TCOMM.2025.3578821}
}

@inproceedings{shen2026hdpc,
  title  = {{HDPC} Codes with {LDPC} Matrices: {Construction} Based on Social Golfer Problem},
  author={Shen, Yifei and {\"U}nal, Hasan Said and Burg, Andreas},
  booktitle={Proc. IEEE Int. Symp. Inf. Theory (ISIT)},
  pages={1--6},
  year={2026}
}
